\documentclass[11pt, a4paper]{article}

\usepackage[margin=1in]{geometry}
\usepackage{amsmath, amssymb, amsthm}
\usepackage{hyperref}
\usepackage{graphicx}
\usepackage{bbm}

\hypersetup{
    colorlinks=true,
    linkcolor=blue,
    filecolor=magenta,      
    urlcolor=cyan,
}

\newtheorem{theorem}{Theorem}[section]
\newtheorem{proposition}[theorem]{Proposition}
\newtheorem{lemma}[theorem]{Lemma}
\newtheorem{corollary}[theorem]{Corollary}
\theoremstyle{definition}

\theoremstyle{remark}
\newtheorem{remark}[theorem]{Remark}

\numberwithin{equation}{section}
\numberwithin{theorem}{section}

\title{\bf Gegenbauer polynomials and fluctuation properties of the one-dimensional Riesz gas}
\author{}
\date{}

\author{Peter J. Forrester}
\date{}

\begin{document}

\maketitle

School of Mathematics and Statistics,  The University of Melbourne,
Victoria 3010, Australia. \: \: Email: {\tt pjforr@unimelb.edu.au}\\

\bigskip

\begin{abstract}
The Riesz gas in one-dimension consists of particles interacting via a pair potential, ${\rm sgn}(s) |x - x'|^{-s}$, $s \ne 0$ and
$-\log | x - x'|$ for $s=0$. In the infinite density limit, with the particle support the interval $[-1,1]$, we apply a functional derivative method due to Beenakker to compute the covariance of two smooth linear statistics for the Riesz gas with
exponent $s \in (-1,1)$, $s \ne 0$. This we give in terms of a sum over Fourier components of the linear statistics with respect to a Gegenbauer polynomial $\{C_n^{(s/2)}(x) \}$ basis, which generalises a known form in the case $s=0$ involving a cosine expansion. For the power sum linear statistic, our general formula can be reduced to a product of gamma function form, and compared against recent exact results in the literature for this case.
\end{abstract}

\begin{center}
{\it Dedicated to the life and achievements on exactly solved models of Rodney James Baxter}
\end{center}

  \setcounter{section}{-1}
 \section{Prologue}
 At the recommendation of my Honours and Masters year thesis supervisor E.R.~Smith, I applied to study for my PhD under Rodney's supervision by writing him a letter to that effect in late 1982. As it turned out, Kenneth Wilson was awarded the Nobel Prize for physics around that time for his contribution to the renormalisation group theory of the phase transitions. Rodney was brought down from Canberra to Melbourne by the Physics School here to present a public lecture on this ---
 indeed the primary motivation for Rodney'a field of exactly solved models in statistical mechanics was in relation to phase transitions, and his discovery of continuously varying critical exponents in the eight-vertex model presented a particular challenge.
 He took the opportunity to arrange an interview with me, which I evidently passed as he announced that he'd be pleased to take on the role as my PhD supervisor. While much could be said about what I observed in regard to his research over the next two years (in particular his celebrated collaboration with George Andrews on what is now often termed the ABF solid-on-solid model
 \cite{ABF84}, involving generalisations of the Rogers-Ramanujan identities, and his exact evaluation of the free energy of the three-dimensional Zamolodchikov model \cite{Ba84}), 
 I'll mention these only in the context of what Rodney gave to me to work on. In addition, I'd like to take this opportunity
 say a few words of a more personal nature.

 From the beginning Rodney (and his wife Elizabeth) treated me like  family, with frequent invitations to his home for dinners, often also in the company of high standing visitors
 (Douglas Abraham, George Andrews, Michael Barber, Larry Glasser, Barry McCoy, Jan Oitmaa, Derek Robinson, David Ruelle, David Thouless,
 $\dots$). Morning tea in the aptly named tea room, served by the now extinct tea lady employees, was an essential part of the day. The period after morning tea and before lunch was no doubt one of much productivity and creativity --- having a ``clear head'' was a self described terminology I remember him using --- typically recorded on dated and numbered notepad paper. Excellence was an academic quality that he put to the fore,  as was understated hard work
 (remember the importance of morning tea),
 and on this he very much led by example. Similar remarks could be made with respect to having grace and good humour.
 
 In his own time as a PhD student at the ANU in Australia, after his well known stint working for the Iraq Petroleum Company at the end of his Cambridge degree, Rodney found for himself a  thesis topic 
 on the statistical mechanics of gases. During  that time, among other things, he worked out the exact solution of the one-dimensional Coulomb plasma (one mobile species, immersed in a neutralising background)
 \cite{Ba63}. While this was topical in mathematical physics, due to recent work of Lenard on the corresponding two-species one-dimensional Coulomb system \cite{Le61}, it
 was in a different direction to the field theory topics suggested to him by his PhD supervisor
 Kenneth Le Couteur. 
 
 It so happened that my MSc thesis from the year before beginning  at ANU related to exactly solvable properties of the one-component plasma in the two-dimensional case (logarithmic potential). In keeping with his own PhD experience, Rodney was only too happy that I continue working on this topic at least some of the time. But for the reality of where a PhD in exact solutions in statistical mechanics may lead --- something that was already highlighted in Rodney's letter responding to my initial PhD inquiry --- he also encouraged work on what he knew as high profile problems in his field of exactly solved lattice models. First I was assigned  to some (relatively minor) theta function identities required for what was soon to become an important work \cite{ABF84}, then I was given the more substantial task of extending this exact solution to 
 include a further parameter \cite{FB85}. Simultaneous there was a numerical project relating to the extension of the three-dimensional Zamolodchikov model \cite{BF85}, based on a three-dimensional extension of Rodney's famed corner transfer matrices \cite[Ch.~13]{Ba82}, crucial to both \cite{ABF84,FB85}. Then I was to learn that this idea had its origin in a much earlier numerical study of Rodney \cite{Ba68}. In his reply letter to my inquiry, Rodney had identified computing skills in relation to flexibility in the job market.
 Exact numerical explorations, performed by running programs written in Fortran, was a common method Rodney used to progress his research.

 Rodney officially retired his ANU position at the end of 2002, just before this 63rd birthday. While there may have been several forces at work in this decision (including his appeal to classical Roman traditions of completion being associated with multiples of 7), one perhaps lesser known one that came up in a tearoom conversation was his respect for one of his Cambridge teachers, the mathematical physicist John Polkinghorne \cite[\S 6]{Ba15}. The latter, after 25 years in science,   decided he had done his bit for the discipline, and took on the challenge of a second career (theology). I speculate that Rodney felt he had done his bit for ANU. In his research, he had for over a decade been toiling at a derivation of the conjectured spontaneous magnetisation formula for chiral Potts model. It struck me that Rodney had retired from his ANU position somewhat defeated by at least this one aspect  of what had become his major research theme (I hasten to add that his contribution generally to the study of chiral Potts model was already substantial at that stage by any measure; see e.g.~the earlier references cited in \cite{Ba06}).
 However, while Rodney had retired from his ANU position, his research continued, and around two years later he indeed found the proof that he had long been  seeking \cite{Ba05a,Ba05b}.

\section{Introduction and paper outline}
The present work addresses  a particular fluctuation problem for the one-dimensional classical statistical mechanical model, specified by a long-rangle Riesz potential
\begin{equation}\label{1.1}
\Phi_s(x,x') = \begin{cases} - | x - x'|^{-s}, & -2 < s < 0 \\
- \log | x - x'|, & s = 0 \\
|x - x'|^{-s}, & 1 > s > 0. 
\end{cases}
\end{equation}
Here the upper bound $1 > s$ separates the cases that the tail of $\Phi_s(x,x')$ is integrable; when this is not integrable is the regime of a long-range potential. The lower bound $ s > - 2$ is required for positivity of the Fourier transform, which in turn is necessary for screening and thus thermodynamic stability,
the latter assuming  too the presence of a neutralising background
\cite{Le22}, \cite[\S{1.2}]{BF25}. In our study, it turns out that the assumption of a neutralising background plays no role. A setting relating to this class of 
Riesz gases where this does play a role is that of determining the global density profile; references on this include
\cite{LS17,HLSS18,A+19,K+21,BFMS26}.

The fluctuation problem to be considered is that of the $N \to \infty$ covariance 
for two smooth and real linear statistics $F_N = \sum_{j=1}^N f(x_j)$,  $G_N = \sum_{j=1}^N g(x_j)$, when the gas is confined to a finite interval
taken to be $[-1,1]$. Established in the literature (see e.g.~the review \cite[\S 3.1]{Fo23}) is that for the logarithmic case $(s=0)$,
\begin{equation}\label{1.2}
\lim_{N \to \infty} {\rm Cov} (F_N,G_N) = {2 \over \beta} \sum_{n=1}^\infty n f_n^{\rm c} g_n^{\rm c}, \quad  f_n^{\rm c} := {1 \over \pi} 
\int_0^\pi f(\cos x) \cos(n x) \, dx,
\end{equation}
and similarly the meaning of $g_n$. A more recent result \cite{FMS23}, \cite{Be23} is that in the case $s=1$
(one-dimensional Coulomb potential)
\begin{equation}\label{1.3}
\lim_{N \to \infty} {\rm Cov} (F_N,G_N) = {1 \over 2 \beta} \int_{-1}^1 f'(x) g'(x) \, dx.
\end{equation}

Let $\rho_{(2),N}^T(x,x')$ denote the truncated two-point correlation function. Use this and the one-body density
$\rho_{(1),N}(x)$  to specify the density-density correlation
\begin{equation}\label{1.3a}
{\sf N}_{(2),N}(x,x') = \rho_{(2),N}^T(x,x') +
\rho_{(1),N}(x) \delta(x-x').
\end{equation}
For finite $N$ this quantity relates to the covariance by the simple formula (see e.g.~\cite[Prop.~2.1]{Fo23})
\begin{equation}\label{1.3b}
{\rm Cov}(F_N,G_N) = \int_{-1}^1 dx \int_{-1}^1 dx' \, f(x) g(x')
{\sf N}_{(2),N}(x,x').
\end{equation}
Using (\ref{1.3b}) as our starting point, our main task then is to calculate the large $N$ form of the smoothed density-density correlation in a form which makes the double integral explicit (the smoothing comes about since
${\sf N}_{(2),N}(x,x')$ is integrated over in both variables).
For this we will make use of the functional derivative formalism due to Beenakker \cite{Be93}, which we
revise in \S \ref{S2.1}. From a mathematical viewpoint, this a physical heuristic, albeit a well founded one which is expected to be exact in the setting of statistical mechanical systems with a long range potential.

As a result of our working, we are able to express the covariance
for a certain range of $s$
 in a form analogous to
(\ref{1.2}). Perhaps surprisingly, the Gegenbauer polynomials play an essential role.

\begin{proposition}\label{P1.1}
Set
\begin{equation}\label{hl}
\quad h_n = {\pi 2^{1 -s} \Gamma(n +s) \over n! (n + s/2) \Gamma^2(s/2) }, \quad 
\lambda_n = { \pi \Gamma(n+ s) \over \Gamma(s) \Gamma(n+1) \cos {\pi s \over 2}}.
\end{equation}
In the case of the pair potential (\ref{1.1}) for $|s| < 1$,
$s \ne 0$,
we have (assuming convergence of the series)
\begin{align}\label{2.3a}
\lim_{N \to \infty} {\rm Cov} (F_N,G_N) & =
{ {\rm sgn}(s)  \over \beta}  \sum_{n=1}^\infty {h_n \over  \lambda_n } f_n^{\rm G}  g_n^{\rm G} \nonumber \\
& =  { {\rm sgn}(s)  \over \beta} {  \Gamma(s) \cos {\pi s \over 2}  \over  2^{s-1} \Gamma^2(s/2)}
\sum_{n=1}^\infty {1 \over (n + s/2) } f_n^{\rm G}  g_n^{\rm G},
\end{align}
where, with $C_n^{(\nu)}(x) $ denoting the Gegenbauer polynomials in standard notation,
\begin{equation}\label{2.3b}
f_n^{\rm G} := {1 \over h_n} \int_{-1}^1 f(x) (1 - x^2)^{(s-1)/2} 
C_n^{(s/2)}(x) \, dx,
\end{equation}
and similarly the meaning of $g_n^{\rm G}$.
\end{proposition}

In comparison to (\ref{2.3a}), the formula for the covariance given in \cite{Be23} is more complicated in appearance.
On the other hand,
it does not require the bounding interval to be symmetric about the origin, although the integrals therein may not be tractable in the general case.
A specific case where they were shown to be tractable
and a closed form obtained was that of
the bounding interval $[0,L]$, and the monomial functions
$f(x) = x^p$ and $g(x) = x^q$, $p,q \in \mathbb Z^+$.
However only for $p=q=1$ is this (take $L=2$) equivalent to the case of the bounding interval equalling $[-1,1]$. 

Suppose we specialise the covariance for the monomials to a variance by setting $p=q$.
Then it has recently been made explicit in \cite[Appendix B]{LS25} that this with the bounding
interval $[-1,1]$ can be written as a linear combination of
covariances with $f(x) = x^{p_1}$ and $g(x) = x^{p_2}$ in the case of the interval $[0,2]$. 
Moreover, using a method distinct from that in \cite{Be23}, this
reference \cite[see Appendix B for an overview]{LS25}
provided an evaluation of the monomial variance formula for 
the case of a symmetric bounding interval for all
$p$ even. For small even $p$ at least, it was checked that the linear combination of the covariances from \cite{Be23} can be be reduced to this same closed form expression, although a general demonstration appears out of reach.

 As an application of (\ref{2.3a}), for the symmetric interval $(-1,1)$ we show in
\S \ref{S2.5} how to reclaim the evaluation formula 
in \cite{Be23} for the monomial variance in the case $p=1$, and
the evaluation formula 
in \cite{LS25} for the monominal variance with $p$ even.  Moreover we obtain an analogous formula in the $p$ odd case.

Our presentation begins in \S \ref{S2.1} with an overview of linear response/ functional derivative strategy of Beenakker for the computation of the $N \to \infty$ smoothed density-density
correlation as appears in the covariance formula (\ref{1.3b}).
Specifically the linear response argument leads to a particular integral equation, (\ref{2.4}) below, the solution to which determines (as a distribution) the smoothed density-density
correlation by functional differentiation. It is shown in
\S \ref{S2.2}  how to solve this in the case $s=0$ of
(\ref{1.1}) (logarithmic potential) and to continue on to deduce
(\ref{1.2}). The case $s=-1$ (linear potential)
of (\ref{2.4}) is shown in
\S \ref{S2.3} to yield to a differentiation strategy based on the one-dimensional Poisson equation. The cases $s \in (-1,1)$,
$s \ne 0$ are considered in \S \ref{S2.4}, where the working begins by deriving an expansion for Riesz potential as a bilinear sum of Gegenbauer polynomials. The use of this to solve the integral equation (\ref{2.4}). After functional differentiation, in Corollary \ref{C2.3} we present a formula for the smoothed
density-density
correlation in terms of further bilinear sum of Gegenbauer polynomials, from which Proposition \ref{P1.1} follows.
In Appendix A 

Before getting underway with this working, we draw to the attention of the reader several other works which relate to fluctuation properties of the Riesz gas of a that have appeared in the recent literature. One is that of a central limit theorem for certain singular linear statistics of the Riesz gas with periodic boundary conditions and $s \in (0,1)$ \cite{Bo21}.
This result includes the some explicit variance formulas. The reference \cite{TLS25} also considers the periodic boundary conditions model, now in what may be considered as the low temperature limit, when the interaction potential can be approximated by a certain multivariate Gaussian, and several results pertaining to fluctuations are obtained. In the specific case of an harmonic confining potential, a study of edge fluctuation properties is given in \cite{K+22}. The work
\cite{DKM23} considers the case of dynamical fluctuations of the Riesz gas due to overdamped Brownian motion.

\section{Computation of the variance for bounding interval $[-1,1]$}
\subsection{The strategy of Beenakker}
\label{S2.1}
In this subsection we will revise the formalism of \cite{Be93} for the computation of the smoothed, large $N$ form of the density-density correlation. We adopt an approach to this formalism using linear response as presented in \cite[\S 14.3.1]{Fo10}.
Introduce the microscopic density $n_{(1),N}(x) := \sum_{j=1}^N \delta(x - x_j)$. Suppose the system is perturbed by the arbitrary one-body potential 
$\delta U = \sum_{j=1}^N u(x_j)$. For an observable $A$, the change in its mean value due to the perturbation is to leading order in $\delta U$ given by (see 
\cite[Eq.~(14.1)]{Fo10})
\begin{equation}\label{2.1}
\langle A \rangle_\epsilon - \langle A \rangle_0 = - \beta 
\langle A \delta U \rangle_0^T.
\end{equation}
Here the subscripts $\epsilon$ and 0 indicate the presence and absence, respectively, of the perturbation.
For the choice $A = n_{(1),N}(x')$, (\ref{2.1}) reads
\begin{equation}\label{2.2}
\langle n_{(1),N}(x') \rangle_\epsilon - \langle n_{(1),N}(x') \rangle_0 = - \beta \int_{-1}^1 u(x) {\sf N}_{(2),N}(x,x') \, dx.
\end{equation}
Functional differentiation then gives
\begin{equation}\label{2.3}
{\delta \over \delta u(x)} \Big ( 
\langle n_{(1),N}(x') \rangle_\epsilon - \langle n_{(1),N}(x') \rangle_0 \Big ) = - \beta {\sf N}_{(2),N}(x,x'),
\end{equation}
so our task is to obtain an expression for the LHS of (\ref{2.2}) that is well suited to computing the required functional derivative.

For this, under the assumption of a long-range potential
$\Phi_s(x,x')$,
and in the large $N$, infinite density limit, it is hypothesised that the density difference as appears in (\ref{2.2}) is such that it couples with $\Phi_s(x,x')$ to effectively cancel the potential $u(x)$ (an equivalent viewpoint is that it can be regarded as a background charge density of opposite sign, with $u(x)$ the resulting potential; see \cite[\S 1.1 and 1.2]{BF25},
\begin{equation}\label{2.4}
-\int_{-1}^1 \Phi_s(x,x') 
\Delta n_{(1),\infty}(x') 
 \, dx' = u(x) + C, \quad \Delta n_{(1),\infty}(x'):= \langle n_{(1),\infty}(x') \rangle_\epsilon - \langle n_{(1),\infty}(x') \rangle_0,
 \end{equation}
 for $x \in (-1,1)$.
Here the constant $C$ is determined by the particle conservation condition
\begin{equation}\label{2.5}
\int_{-1}^1 \Delta n_{(1),\infty}(x')  \, dx' = 0.
\end{equation}
Thus the task at hand is to solve (\ref{2.4}) for
the density difference, given $u(x)$. Again we emphasise that the solution has to be in a form that facilitates the computation of the functional derivative.

\subsection{Implementation for $s=0$}\label{S2.2}
The case $s=0$ in (\ref{1.1}) is a distinct functional form from that of the cases $s \ne 0$. In this case the solution of (\ref{2.4}) has been given in \cite[Exercises 14.3 q.1]{Fo10}. It is instructive to revise the required working, and to show how this leads to (\ref{1.2}) (in the cited reference, a variation on (\ref{1.2}) is deduced;
see the first expression in
\cite[Eq.~(14.56)]{Fo10}).

\begin{proposition}\label{P2.1}
Let $x \in (-1,1)$.
The solution of the integral equation
\begin{equation}\label{N1}
\int_{-1}^1 \log | x - x'|  \Delta n_{(1),\infty}(x')  \, dx' =
u(x) + C,
\end{equation}
where $C$ is such that (\ref{2.5}) holds, is given by the cosine expansion
\begin{equation}\label{2.6}
 \Delta n_{(1),\infty} (\cos \theta) = - {2 \over \pi^2 \sin \theta}
 \sum_{p=1}^\infty p \Big ( \int_0^\pi u (\cos \sigma) \cos p \sigma \, d \sigma \Big ) \cos p \theta.
\end{equation}
\end{proposition}

\begin{proof}
The substitutions $x = \cos \theta$, $x' = \cos \sigma$, gives the rewrite of (\ref{N1}),
\begin{equation}\label{2.6a}
\int_{0}^\pi \log | \cos \theta  - \cos \sigma|  \delta n_{(1),\infty}(\cos \sigma)  \sin \sigma \, d\sigma =
u(\cos \theta) + C.
\end{equation}
In relation to the kernel on the LHS, we have the cosine
expansion \cite{PS90}, \cite[Exercises 1.4 q.4]{Fo10}
\begin{equation}\label{2.6b}
 \log(2 | \cos \theta  - \cos \sigma|) = - 
 \sum_{p=1}^\infty {2 \over p} \cos p \theta \cos p \sigma.
\end{equation}

Recalling the condition (\ref{2.5}) we see that the RHS of
(\ref{2.6b}) can be substituted for the kernel in (\ref{2.6a}). Doing this, performing a cosine expansion of $u(\cos \theta)$, making further use of the condition (\ref{2.5}), and equating coefficients of $\cos p \sigma$ shows that for each $p=0,1,\dots$
\begin{equation}\label{2.6c}
\int_0^\pi (\cos p \sigma ) \Delta n_{(1),\infty}(\cos \sigma)  \sin \sigma \, d\sigma = - {p \over \pi} 
\int_0^\pi u (\cos \sigma) \cos p \sigma \, d \sigma.
\end{equation}
Apart from a factor of $2/\pi$, the LHS of (\ref{2.6c}) are the coefficients in the cosine expansion of $\Delta n_{(1),\infty}(\cos \sigma)  \sin \sigma$. Replacing these coefficients according to (\ref{2.6c}) gives (\ref{N1}).
    
\end{proof}

Changing variables as in (\ref{2.6a}) in (\ref{2.3})
allows us to compute from (\ref{2.6c}) that
\begin{equation}\label{2.6d}
\beta {\sf N}_{(2),\infty}(\cos \theta,\cos \sigma) =
- {1 \over \pi^2 \sin \theta \sin \sigma}
{\partial^2 \over \partial \theta^2} 
\log | \cos \theta - \cos \sigma |,
\end{equation}
where use has been made of (\ref{2.6b}). Similarly changing variables in (\ref{1.3b}) with $N \to \infty$, we see by substituting (\ref{2.6d}) and making further use of (\ref{2.6b}) that
the fluctuation formula (\ref{1.2}) results.

\subsection{The case $s=-1$}\label{S2.3}
Recalling (\ref{1.1}), the case $s=-1$ of (\ref{2.3}) reads
\begin{equation}\label{2.7}
\int_{-1}^1 | x - x'| \Delta n_{(1), \infty}(x') = u(x) + C.
\end{equation}
Using the fact that ${\partial^2 \over \partial x^2} | x - x'| = 2 \delta(x-x')$, it follows that for $x \in (-1,1)$,
$\Delta n_{(1), \infty}(x') = {1 \over 2} {d^2 \over d x^2} u(x')$ and consequently
\begin{equation}\label{2.8}
{\delta \Delta n_{(1), \infty}(x') \over \delta u(x)} =
{1 \over 2} {d^2 \over d x^2} \delta (x - x').
\end{equation}
Substituting this in (\ref{1.3b}) with $N \to \infty$, we
reclaim (\ref{1.3}).

\subsection{Proof of Proposition \ref{P1.1}}\label{S2.4}
Our main task is to solve the integral equation (\ref{2.4})
with the Riesz potential (\ref{1.1}) in a form suitable to taking the functional derivative required in (\ref{2.3}). 
Let $h_n, \lambda_n$ be specified as in (\ref{hl}).
With $C_n^{(\alpha)}(x)$ denoting the Gegenbauer polynomials in usual notation, which have the orthogonality
\begin{equation}\label{2.9}
\int_{-1}^1 (1 - y^2)^{(s - 1)/2}
C_n^{(s/2)}(y) C_m^{(s/2)}(y) \, dy =
h_n \delta_{n,m}. 
\end{equation}
Crucial to our working is
a particular integral operator eigenvalue equation for these polynomials
known  in the literature  \cite{MA90},
\cite[Eq.~(6.2)]{FAD99} (see \cite[\S 7]{Ab01} for a self contained derivation), giving that for $|s| < 1$, and $|u| \le 1$,
\begin{equation}\label{2.10}
\int_{-1}^1  {(1 - y^2)^{(s - 1)/2} \over| u - y|^s} 
C_n^{(s/2)}(y) \, dy = \lambda_n C_n^{(s/2)}(u). 
\end{equation}
The LHS of (\ref{2.10}) is, up to proportionality, equal to the Fourier coefficients for the orthogonal basis
$\{  C_n^{(s/2)}(y) \}$, inner
product $(f_1,f_2) := \int_{-1}^1 (1 - y^2)^{(s - 1)/2} f_1(y) f_2(y) \, dy$,
 of the
function $| u - y|^{-s}$ in the variable $y$. Hence
(\ref{2.10}) is equivalent to the expansion
\begin{equation}\label{2.11a}
{1 \over | u - y|^s} = 
\sum_{n=0}^\infty
{\lambda_n \over h_n}  C_n^{(s/2)}(u) C_n^{(s/2)}(y), \quad |u|,|y|, \le 1, \: \: |s| < 1;
\end{equation}
cf.~(\ref{2.6b}).
In Appendix A, a further discussion relating to (\ref{2.10}) is given, which includes a self contained derivation.

We can use (\ref{2.11a}) to deduce the analogue of 
Proposition \ref{P2.1} in the case of the potentials
in (\ref{1.1}) with $|s| < 1$, $s \ne 0$.

\begin{proposition}\label{P2.1a}
Let $x \in (-1,1)$ and $|s| <1$, $s \ne 0$.
The solution of the integral equation
\begin{equation}\label{U1}
- {\rm sgn}(s) \int_{-1}^1 {1 \over  | x - x'|^s}  \Delta n_{(1),\infty}(x')  \, dx' =
u(x) + C,
\end{equation}
where $C$ is such that (\ref{2.5}) holds, is given by
\begin{multline}\label{U2}
 \Delta n_{(1),\infty}(x) = - {\rm sgn}(s)
 \\
\times (1 - x^2)^{(s - 1)/2}
 \sum_{n=1}^\infty{1 \over h_n \lambda_n} \Big (
 \int_{-1}^1 (1 - y^2)^{(s-1)/2} C_n^{(s/2)}(y) u(y) \, dy \Big )  C_n^{(s/2)}(x).
\end{multline}
\end{proposition}

\begin{proof}
    We expand $u(x)$ in terms of $\{ C_n^{(s/2)}(x) \}$,
  \begin{equation}\label{2.11b}  
   u(x) = \sum_{n=0}^\infty {1 \over h_n} 
  \Big (  \int_{-1}^1 (1 - y^2)^{(s-1)/2} C_n^{(s/2)}(y) u(y) \, dy \Big ) C_n^{(s/2)}(x).
 \end{equation} 
 After substituting both (\ref{2.11a}) and (\ref{2.11b}) in
 (\ref{U1}) and recalling the requirement (\ref{2.5}), equating coefficients of $ C_n^{(s/2)}(x)$ shows that for
 $n=1,2,\dots$
  \begin{equation}\label{2.11c}  
\lambda_n  \int_{-1}^1  C_n^{(s/2)}(y) 
\Delta n_{(1),\infty}(y) \, dy =\int_{-1}^1 (1 - y^2)^{(s-1)/2} C_n^{(s/2)}(y) u(y) \, dy.
\end{equation}

On the other hand, substituting $ (1 - x^2)^{-(s-1)/2}
\Delta n_{(1),\infty}(x)$ for $u(x)$ in
(\ref{2.11b}), and making use of (\ref{2.5}) gives
 \begin{equation}\label{2.11d}  
  (1 - x^2)^{-(s-1)/2}
\Delta n_{(1),\infty}(x)   = \sum_{n=1}^\infty {1 \over h_n} 
  \Big (  \int_{-1}^1  C_n^{(s/2)}(y) \Delta n_{(1),\infty}(y) \, dy \Big ) C_n^{(s/2)}(x).
 \end{equation}
 Using (\ref{2.11c}) to substitute for the integral on the RHS gives 
 (\ref{U2}).
  
\end{proof}

The computation of the functional derivative with respect to $u(y)$ as required by
(\ref{2.3}) is immediate.

\begin{corollary}\label{C2.3}
In the case of the pair potential $|x-y|^{-s}$ for $|s| < 1$,
$s \ne 0$,
we have
 \begin{equation}\label{N3}
\beta {\sf N}_{(2),\infty}(x,y) = {\rm sgn}(s) (1 - x^2)^{(s-1)/2}
(1 - y^2)^{(s-1)/2} \sum_{n=1}^\infty {1 \over h_n \lambda_n}
C_n^{(s/2)}(x) C_n^{(s/2)}(y).
 \end{equation}
 (This holds in the sense of distribution, and needs to be integrated against test functions to be well defined.)
\end{corollary}

The covariance formula of Proposition \ref{P1.1} follows immediately from (\ref{N3}), upon substitution in (\ref{1.3b}).

\begin{remark}
The integral equation (\ref{U1}), with ${1 \over 2} u(x)$ interpreted as an external potential and $\Delta n_{(1),\infty}(x)$ as the global density $\rho_{(1)}^{\rm G}(x)$, is
the Euler-Lagrange equation for the latter; see
e.g.~\cite[Eq.~(1.7)]{BFMS26}. One notes now the normalisation
$\int_{-1}^1 \rho_{(1)}^{\rm G}(x') \, dx'$ rather than the requirement (\ref{2.5}). In the case $u(x) = 0$, and the particles confined to $(-1,1)$ (box wall potential), the method of solution leading to (\ref{U2}) gives
 \begin{equation}\label{N3b}
 \rho_{(1)}^{\rm G}(x) \propto (1 - x^2)^{(s-1)/2}
 \mathbbm 1_{|x|<1};
\end{equation}
see \cite[Remark 4]{BFMS26} for more on this.

\end{remark}

\subsection{The variance of the power sum linear statistic}\label{S2.5}
Let us denote
 \begin{equation}\label{J1}
 \sigma_p^2 := \lim_{N \to \infty} {\rm Cov}(F_N,G_N) \Big |_{f = g = x^p},
 \end{equation}
 which is the limiting variance of the linear statistic given by the power
sum $\sum_{j=1}^N x^p$. Since $x^p$ can be written as a linear combination of $\{ C_n^{(s/2)}(x) \}_{n=0}^p$, it follows that only the first $p$ terms in (\ref{2.3a}) contribute to $\sigma_p^2$ --- in fact this is further reduced by a factor of 2 due to parity considerations.

The simplest case in this class is $p=1$, for which the power sum linear statistic has the interpretation as the centre of mass. Since $x = {1 \over s}  C_1^{(s/2)}(x)$, we read off from
(\ref{2.3a}) that
 \begin{equation}\label{J2}
 \sigma_1^2 = {{\rm sgn}(s) \over s^2 \beta} {h_1 \over \lambda_1} = {{\rm sgn}(s) \over s^2 \beta} {\cos (\pi s/2)
 \Gamma((s+1)/2) \over \sqrt{\pi} (1+s/2) \Gamma(s/2)},
 \end{equation}
 where use has been make of the duplication formula for the gamma function.
 This same quantity has been computed in \cite[Eq.~(4.1) with $L=2$,
 $N^2 J = 1$ --- the translational invariance of $\sigma_1^2$ makes only the length of the bounding interval relevant, not its position on the line]{Be23}. After use of the reflection identity for the gamma function, the formula given there, and our formula, are seen to be identical. It is noted in \cite{Be23} that $|s| \sigma_1^2$ as specified by (\ref{J2})
 is analytic and positive throughout the exponent range $1 > s > -2$, which led to the conjecture of its correctness for $-1 > s > -2$ in (\ref{1.1})
 (recall that the validity of Proposition \ref{P1.1} requires
 $|s| < 1$, $s \ne 0$).

 Quite remarkably, the general formula of \cite{Be23} corresponding to our (\ref{2.3a}), notwithstanding its complexity, was shown 
 \cite[Eq,~(3.4)]{Be23}
 to also give rise to a gamma function evaluation of $\sigma_p^2$ in the case $p > 1$. However, as discussed in the Introduction,  an important detail is that the bounding interval is taken to be $[0,L]$ which prohibits direct comparison with the value of $\sigma_p^2$ obtained from
 our (\ref{2.3a}). Nonetheless, as also  discussed in the Introduction,
 direct comparison is possible in the case
  $p$ even from a result presented in \cite{LS25}. In our notation, and after some minor manipulation involving the gamma function duplication identity, this result reads
\begin{equation}\label{J3a}
\sigma_p^2 = \sigma_2^2 \, 2^{2(p-2)} \Big ( {p \over 2} \Big )^3
\Big ( {2 + s \over p + s} \Big ) \alpha_p(s), \quad p \: \:{\rm even},
\end{equation}
where
\begin{equation}\label{J3b}
\sigma_2^2 = {1 \over  \beta \sqrt{\pi}}
  {\cos {\pi s \over 2} \Gamma((s+1)/2) \over |s| (2+s)  \Gamma({s \over 2} + 3)},
\end{equation}
and
\begin{equation}\label{J3c}
\alpha_p(s) = {2^{-2p + 8} \Gamma^2((p+1)/2) \Gamma({s \over 2} + 1) \Gamma({s \over 2} + 3)  \over
\pi p (2p + s) \Gamma({p+s \over 2} + 1)  \Gamma((p+s) / 2 ) }.
\end{equation}


We have the challenge then of obtaining (\ref{J2}) and
(\ref{J3a}) starting with (\ref{2.3a}), and further to obtain
an evaluation formula for $\sigma_p^2$ for $p>1$ odd.
To accomplish this, we require the expansion of a monomial in terms of Gegenbauer polymomials
  \cite{Ra60}
  \begin{equation}\label{J4} 
  x^p = {p! \over 2^p} \sum_{k=0}^{[p/2]}
  {\nu + p - 2k \over k! (\nu)_{p+1-k}}
  C_{p-2k}^{(\nu)}(x),
   \end{equation}
   where $(a)_k$ denotes the rising Pocchammer symbol.
   Substituting this with $\nu = s/2$ in (\ref{2.3a}) gives
  \begin{equation}\label{J5}
 \sigma_p^2 = { {\rm sgn}(s)  \over \beta} {  \Gamma(s) \cos {\pi s \over 2}  \over  2^{s-1} \Gamma^2(s/2)}
\Big ( { p! \over 2^p } \Big )^2 \sum_{k=0}^{[(p-1)/2]} 
  {s/2 + p - 2k \over (k! (s/2)_{p+1-k})^2}
   \end{equation}
   (the expression for the upper terminal in the summation has been changed relative to (\ref{J4}) to account for the term
   $n=0$ not being present in (\ref{2.3a})).

   The following result shows that the summation in (\ref{J5}) can be evaluated.
\begin{lemma}
For $|s|<1$, $s \ne 0$ we have
 \begin{equation}\label{J6}
\sum_{k=0}^{n} 
  {s/2 + p - 2k \over (k! (s/2)_{p+1-k})^2} =
  {1 \over (n!)^2} {1 \over s/2 + p}
  {1 \over ((s/2)_{p-n})^2}.
\end{equation}
\end{lemma}

\begin{proof}
    We see that the proposed formula is correct for $n=0$.
    Proceeding now by induction, we have that the sum for $n \mapsto n+1$ is equal to
    \begin{multline}
    {1 \over (n!)^2} {1 \over s/2 + p}
  {1 \over ((s/2)_{p-n})^2} +  
  {s/2 + p - 2(n+1) \over ((n+1)! (s/2)_{p-n})^2} \\ =
    {1 \over (n!)^2} {1 \over ((s/2)_{p-n})^2} \bigg (
  {1 \over s/2 + p} +    {s/2 + p - 2(n+1) \over (n+1)^2}
  \bigg ).
    \end{multline}
    Simplification of the final factor on the RHS gives agreement
    with the RHS of (\ref{J6}) for $n \mapsto n+1$, as required.
\end{proof}
   
Use of (\ref{J6}) in (\ref{J5}) shows that our formalism, as for that of \cite{Be23}, gives a product of gamma function evaluation of $\sigma_p^2$. Due to the appearance of
$[(p-1)/2]$ there is need to distinguish the cases $p$ odd and $p$ even.

\begin{corollary}\label{C2.5}
 We have
 \begin{equation}\label{J7a}
 \sigma_p^2 = {{\rm sgn}(s) \over \pi \beta}  {\Gamma(s) \cos {\pi s \over 2} \over 2^{s-2} (s + 2p) } \bigg (
 {\Gamma({p \over 2} + 1) \over \Gamma((s+p+1)/2) } \bigg )^2,
  \quad p \: \:{\rm odd},
 \end{equation}
 and
 \begin{equation}\label{J7b}
 \sigma_p^2 = {{\rm sgn}(s) \over \pi \beta}  {\Gamma(s) \cos {\pi s \over 2} \over 2^{s-2} (s + 2p) } \bigg ( 
 {(p/2) \Gamma(( p  + 1)/2)  \over \Gamma({s+p \over 2} + 1) } \bigg )^2,
  \quad p \: \:{\rm even}.
 \end{equation}
\end{corollary}

Straightforward manipulation of the expression (\ref{J3a}) from
\cite{LS25} shows that it can be written as in (\ref{J7b}), thus showing that our formula (\ref{2.3a}) reproduces this result from the existing literature. We have previously noted via
(\ref{J2}) that the case $p=1$ of (\ref{J7a}) reproduces a known result from \cite{Be23}. For $p>1$ the result (\ref{J7a}) is new.

As noted in both \cite{Be23} and \cite{LS25}, taking the limit $s \to -1$ provides a useful check, as then we have the alternative, very simple formula (\ref{1.3}). This latter formula tells us
 \begin{equation}\label{J8}
 \sigma_p^2 \Big |_{s=-1} = {1 \over \beta } {p^2 \over 2 p -1}.
  \end{equation}
  Indeed taking the limit   $s \to -1$ in both (\ref{J7a}) and
  (\ref{J7b}) gives precisely this result.

  \begin{remark} ${}$ \\
  1.~The leading large $p$ form of (\ref{J8}) is by inspection
  $p/(2 \beta)$. Beyond the case $s=-1$,, we have from
  Corollary \ref{C2.5} that the leading large $p$ form is given by
   \begin{equation}\label{J9}
   {{\rm sgn}(s) \cos {\pi s \over 2} \over 4 \pi \beta} 
   \Big ( {2 \over p} \Big )^s.
   \end{equation} 
   This diverges for $s<0$, and goes to zero for $s > 0$. 
   One notes that for $p \to \infty$ the power sum linear statistic is concentrated about $x = \pm 1$, so this result suggests two distinct edge behaviours depending on the sign of $s$. \\
   2.~The variance of the linear statistic $f(x) = \Phi_s(x,y)$ as given by (\ref{1.1}) for $|s|<1$ ($s \ne 0$) and $y \in [-1,1]$ also exhibits contrasting behaviours depending on
   the sign of $s$. According to (\ref{2.3a}) and
   (\ref{2.11a}) we have 
   \begin{align}\label{2.3aE}
\lim_{N \to \infty} {\rm Var} (F_N) & =
{ {\rm sgn}(s)  \over \beta}  \sum_{n=1}^\infty {\lambda_n \over  h_n } ( C_n^{(s/2)}(y))^2 \nonumber \\
& =  { {\rm sgn}(s)  \over \beta} {  \Gamma(s) \cos {\pi s \over 2}  \over  2^{s-1} \Gamma^2(s/2)}
\sum_{n=1}^\infty (n+s/2) ( C_n^{(s/2)}(y))^2.
\end{align}
There are explicit gamma function formulas available for
$C_n^{(s/2)}(y)$ at $y=\pm 1,0$. From these we have that 
(\ref{2.3aE}) diverges for $s> 0$, and converges for $s<0$. Note that for $s \to -1$ this latter conclusion is consistent with
(\ref{1.3}). which gives $\lim_{N \to \infty} {\rm Var} (F_N) =
1/\beta$, independent of $y$. Also, for the case $s=0$ of
(\ref{1.1}) as the linear statistic, and when the Riesz gas is the log-gas, it is well known that the variance diverges for
$y \in [-1,1]$ \cite{DL14,SF25}. \\
   3.~As with (\ref{J2}), the general $p$ expressions
   (\ref{J7a}) and (\ref{J7b}) multiplied by $|s|$ are analytic and positive throughout the entire exponent range $1 > s > -2$, which leads to the conjecture of a corresponding extension of their validity beyond that stated in the corollary.

   \end{remark}

   \section*{Acknowledgements}
The work of the author is supported by a grant from the Australian Research Council, Discovery Project
DP250102552.

   \appendix
\section*{Appendix A: On the integral operator eigenvalue equation (\ref{2.10})}
\renewcommand{\thesection}{A} 
\renewcommand{\theHsection}{A} 
\setcounter{equation}{0}
\setcounter{theorem}{0}

Introduce the differential operator
 \begin{equation}\label{K1}
 \mathcal L = (1 - x^2) {d^2 \over d x^2} - (s+1) x {d \over d x}.
 \end{equation}
 As is well known, up to normalisation, the Gegenbauer polynomials $C_n^{(s/2)}(x)$ are the unique polynomial eigenfunctions of the eigenvalue equation
 \begin{equation}\label{K2}
 \mathcal L \phi_n(x) = - n (n + s) \phi_n(x), \quad n \in \mathbb N_0.
  \end{equation}
  From the theory of the so-called Gegenbauer fuctions (second kind solutions of the differential equation (\ref{K2})
  \cite{DFS76}, an alternative characterisation is that
  $C_n^{(s/2)}(x)$ is the unique (up to normalisation) solution of
  (\ref{K2}) which is bounded on $[-1,1]$.

  Now introduce the integral operator
  \begin{equation}\label{K3}
  \mathcal K[f](u) := \int_{-1}^1  {(1 - y^2)^{(s - 1)/2} \over| u - y|^s} 
f(y) \, dy, \quad |u| \le 1.
 \end{equation}
 In the special case $s=1/2$, in the context of a mathematical physics problem relating to the impenetrable Bose gas in one-dimension with Dirichlet boundary conditions, the commutation identity between the differential and integral operators $[\mathcal L, \mathcal K]=0$ was established. Since commuting operators have the same eigenfunctions, this was used to deduce the validity of (\ref{2.10}) 
 \cite{FFG03}
 (the then very recently established result of
 \cite{FAD99}
 giving (\ref{2.10}) for all $|s|<1$ was unknown to the authors
 of \cite{FFG03}, one of whom is author of the present work), from which the case $g=1/2$ of (\ref{2.11a}) was deduced.
 One notes that commuting differential and integral operators appeared long before in mathematical physics in the context of
 random matrix theory \cite{Ga61}.
 In this Appendix, the aim is to give a derivation of 
 (\ref{2.10}) based on theory relating to a further topic in mathematical physics, namely that of Schr\"odinger operators of Calogero-Sutherland as characterised by a one on distance squared pair potential (see e.g.~\cite[Ch.~11]{Fo10}).
 For this we follow ideas put forward in the thesis
 \cite{At16}.

 Define $\psi_0(x) = (1 - x^2)^{(s-1)/4}$, and introduce the conjugated differential operator $\mathcal L^{\rm c} := \psi_0(x)
 \mathcal L \psi_0^{-1}(x)$. In the variable $x = \cos \theta$,
 $0 < \theta < \pi$, this has the simple Calogero-Sutherland type (in a single variable) Schr\"odinger operator form
 \begin{equation}\label{K1a}
 \mathcal L^{\rm c} \Big |_{x = \cos \theta} := 
 \mathcal L^{\rm c}_\theta = 
 - {d^2 \over d \theta^2} + {s \over 2} \Big ( {s \over 2} - 1 \Big ) {1 \over\sin^2 \theta}.
 \end{equation}
 It follows from  (\ref{K2})
 that the square integrable eigenfunctions (wave functions) of (\ref{K1a})
 are 
  \begin{equation}\label{K1b}
  \psi_n(\theta) = (\sin \theta)^{s/2} 
  C_n^{(s/2)}(\cos \theta), \quad s > - 1,
  \end{equation} 
  with corresponding eigenvalues
  \begin{equation}\label{K1c}
  \nu_n(s) := (s/2)^2 + n (n+s).
  \end{equation}

  An observation of \cite{At16} is the so-called kernel function identity
   \begin{equation}\label{K1d}
   (  \mathcal L^{\rm c}_\theta  -  \mathcal L^{\rm c}_\phi)
   K(\theta, \phi;g) 
  =0, \quad 
K(\theta, \phi;g) =  {(\sin \theta \sin \phi)^{s/2} \over | \cos \theta - 
   \cos \phi |^s},
   \end{equation}
   which can be verified directly. We can use this fact to give a derivation of (\ref{2.10}).

   \begin{proposition}
     Let $\lambda_n$ be as in (\ref{hl}). For $|s| < 1$ we have
       \begin{equation}\label{K1e}
   C_n^{(s/2)}(\cos \phi) = {\lambda_n}
   \int_0^\pi {1 \over |\cos \theta - 
   \cos \phi |^s} (\sin \theta)^g C_n^{(s/2)}(\cos \theta) \, d \theta.
   \end{equation} 
   (This is (\ref{2.10}) in trigonometric form.)
   \end{proposition}
   
   \begin{proof}
   For $|s| < 1$ we can use
   $\{\psi_n(\theta) \}$ to expand
   \begin{equation}\label{K1f} 
 K(\theta, \phi;g) = \sum_{n=0}^\infty
   u_n(\phi) \psi_n(\theta), \quad   u_n(\phi) =
 {(\sin \phi)^{s/2} \over h_n} \int_0^\pi {1 \over |\cos \theta - 
   \cos \phi |^s} (\sin \theta)^g C_n^{(s/2)}(\cos \theta) \, d \theta
 \end{equation} 
 (the condition $s > -1$ is required in relation to (\ref{K1b}),
 and the condition $s<1$ is required for the integral in the definition of $u_n(\phi)$ to be well defined).
 Applying (\ref{K1b}) to this tells us that
 \begin{equation}\label{K1g}
  \mathcal L^{\rm c}_\phi  \Big (  (\sin \phi)^{s/2}  u_n(\phi) 
  \Big )=  \nu_n(s)    u_n(\phi).
  \end{equation}  
  Up to proportionality, $(\sin \phi)^{s/2}C_n^{(s/2)}(\cos \phi)$ is the unique solution of this equation such that $u_n(\phi)/(\sin \phi)^{s/2}$
is bounded.
Hence we have  established (\ref{K1e}), up to the value of $\lambda_n$.
For this latter task, we return to the form 
(\ref{2.10}). Setting $u=1$ gives
 \begin{equation}\label{K1h}
\int_{-1}^1  {(1 - y^2)^{(s - 1)/2} \over( 1 - y)^s} 
C_n^{(s/2)}(y) \, dy = \lambda_n {\Gamma(n+ s) \over n! \Gamma(s)},
\end{equation}
where we have made use of the explicit value of $C_n^{(s/2)}(u)$
evaluated at unity. It turns out that the evaluation of an integral including the one in (\ref{K1h}),  in a product of gamma function form, is given in \cite[Eq,~(7.311.3)]{GR07}.
Minor manipulation gives the expression (\ref{hl}) for $\lambda_n$.

\end{proof}

\end{document}